\documentclass[12pt]{article}
\usepackage[cp1251]{inputenc}
\usepackage{amssymb,amsmath}
\usepackage[dvips]{graphicx,color}
\usepackage{mathrsfs}
\usepackage{array, amsfonts, mathrsfs}
\usepackage{amssymb}
\usepackage{amsthm}
\usepackage{graphics, graphicx}
\usepackage{epstopdf}
\usepackage[colorlinks=true]{hyperref}
\usepackage{float}
\usepackage{enumitem}

\newtheorem{Lemma}{Lemma}
\newtheorem{Theorem}{Theorem}

\newtheorem{Proposition}{Proposition}

\def\blue{\color{blue}}
\def\red{\color{red}}

\date{}

\begin{document}

\title{\blue\bf N-transform and factorization of the DN-map}
\author{\red Mikhail I. Belishev\thanks{St. Petersburg Department of Steklov Mathematical Institute
of Russian Academy of Sciences, Fontanka 27, St. Petersburg,
Russia, 191023; belishev@pdmi.ras.ru},\,\,\, Dmitrii V.
Korikov\thanks{St. Petersburg Department of Steklov Mathematical
Institute of Russian Academy of Sciences, Fontanka 27, St.
Petersburg, Russia, 191023; thecakeisalie@list.ru}.}
\maketitle

\begin{abstract}
Let $(\Omega,g)$ be a smooth compact 3D Riemannian manifold with
the smooth boundary $\Gamma$, $\tau(x):={\rm dist\,}(x,\Gamma)$,
$x\in\Omega$; $\Omega^\tau:=\{x\in\Omega\,|\,\,{\rm
dist\,}(x,\Gamma)<\tau$\}, $\Gamma^\tau:=\{x\in\Omega\,|\,\,{\rm
dist\,}(x,\Gamma)=\tau$\}, $\tau\geqslant 0$. For the sake of
technical simplicity, we deal with $\Omega$ diffeomorphic to a
ball in $\Bbb R^3$.

Let $\mathscr P:=\{\nabla p\,|\,\,p\in H^1(\Omega)\}$ be the space
of the potential vector fields, and let $\mathscr
L_\lambda:=\{\varkappa\nabla\tau\,|\,\,\varkappa\in L_2(\Omega)\}$
be the space of the vector fields parallel to $\nabla\tau$. The
N-transform is a map from $\mathscr P$ to $\mathscr L_\lambda$
defined layer-wise (in accordance with $\Omega=\cup_{\tau\geqslant
0}\Gamma^\tau$) by
$$
Nh\,\big|_{\Gamma^\tau}:=(P^\tau h)\big|_{\Gamma^{\tau-0}},
\qquad\tau>0,
$$
where $P^\tau$ are the projections in $\mathscr P$ onto the
subspaces $\mathscr P^\tau:=\{h\in\mathscr P\,|\,\,{\rm
supp\,}h\subset\overline{\Omega^\tau}\}$. We show that $N$ is a
unitary operator.

Let $p=p^f(x)$ be a solution to the Dirichlet problem: $\Delta_g
p=0$ in $\Omega\setminus\Gamma$, $p=f$ on $\Gamma$. The DN-map
$\Lambda$ is defined by $\Lambda f:=-\langle\nabla
p^f,\nabla\tau\rangle$ on $\Gamma$. We show that the N-transform
provides a certain factorization $\Lambda^{-1}=V^*V$ and discuss
its possible usefulness for determination of $(\Omega,g)$ from
$\Lambda$.

\end{abstract}

\subsubsection*{\blue About the paper}
\smallskip

\noindent$\bullet$\,\,\,The N-transform was introduced in \cite{B
MN transform} in parallel and by analogy with the M-transform,
which plays a key role in the 3D dynamical (time-domain) inverse
problem for the Maxwell system \cite{BGlas_A&A_Maxwell,B IPI
2022}. These transforms correspond to the Helmholtz-Weyl
decomposition of the 3D vector fields into gradients and curls.
However, unlike M-transform, the N-transform has not yet found
application in inverse problems. In our paper, we study its
fundamental properties: isometry and completeness, and discuss
such a possible application in the form of a conjecture related to
a factorization of the elliptic DN-map (the Calderon operator) of
$\Omega$.
\smallskip

\noindent$\bullet$\,\,\,Let $(\Omega,g)$ be a smooth compact 3D
Riemannian manifold with the smooth boundary $\Gamma$,
$\tau(x):={\rm dist\,}(x,\Gamma)$, $x\in\Omega$. We denote
$\Omega^\tau:=\{x\in\Omega\,|\,\,\tau(x)<\tau\}$,
$\Gamma^\tau:=\{x\in\Omega\,|\,\,\tau(x)=\tau\}$,
$\tau\geqslant0$, and represent $\Omega=\cup_{\tau\geqslant
0}\Gamma^\tau$.

Let $\mathscr P:=\{h=\nabla p\,|\,\,p\in H^1(\Omega)\}$ be the
space of the potential fields; let $\mathscr
L_\lambda:=\{v=\varkappa\nabla\tau\,|\,\,\varkappa\in
L_2(\Omega)\}$ be the space of the vector fields longitudinal
w.r.t. $\nabla\tau$. The N-transform is an isometry from $\mathscr
P$ to $\mathscr L_\lambda$. It is constructed layer-wise (in
accordance with $\Omega=\cup_{\tau\geqslant 0}\Gamma^\tau$) via
the breaks $P^\tau h\big|_{\Gamma^{\tau-0}}$ of the projections on
the subspaces $\mathscr P^\tau:=\{h\in\mathscr P\,|\,\,{\rm
supp\,}h\subset\overline{\Omega^\tau}\}$:
$$
Nh\,:=\,(P^\tau h)\big|_{\Gamma^{\tau-0}}\qquad {\rm
on}\,\,\,\Gamma^\tau,\,\,\,\tau>0
$$
(see \cite{B MN transform}). For the sake of technical simplicity,
we deal with $\Omega$ diffeomorphic to a ball in $\Bbb R^3$ and
show that $N$ is a unitary operator.
\smallskip

\noindent$\bullet$\,\,\,Let $\mathscr H:=\{h\in\mathscr
P\,|\,\,h=\nabla p,\,\,\,\Delta
p=0\,\,\,\text{in}\,\,\,\Omega\}$\,\,($\Delta$ the
Beltrami-Laplace operator) be the subspace of the harmonic
potential fields. We show that the fields $N\mathscr
H\subset\mathscr L_\lambda$ are of the form
$v=\varkappa\nabla\tau$, where $\varkappa$ satisfies a first-order
evolutionary differential equation in $\Gamma\times(0,T)$ (in the
semi-geodesic coordinates $(\gamma,\tau)$ with the base on
$\Gamma$) and the Cauchy data
$\varkappa\big|_{\Gamma}=\varkappa_0$. By $V:\varkappa_0\mapsto
\varkappa\big|_{\Gamma\times[0,T]}$ we denote the operator
(propagator) that solves this Cauchy problem.
\smallskip

\noindent$\bullet$\,\,\,The DN-map $\Lambda$ is associated with
the elliptic Dirichlet problem
$$
\begin{cases}
\Delta p=0 & {\rm in}\,\,\,\Omega\setminus\Gamma;\\
p=f &{\rm on}\,\,\,\Gamma,
\end{cases}
$$
with a solution $p=p^f(x)$, and is defined by
$$
\Lambda f\,:=\,\partial_\nu p^f =-\langle\nabla
p^f,\nabla\tau\rangle\qquad {\rm on}\,\,\, \Gamma,
$$
where $\nu$ is the outward normal at the boundary.

We show that the N-transform provides a factorization
$\Lambda^{-{1}}\,=\,V^*V$ and discuss its possible usefulness for
the electric impedance tomography problem, which is determination
of $(\Omega,g)$ from $\Lambda$.
\smallskip

\noindent$\bullet$\,\,\,{\red The work was supported by the
Ministry of Science and Higher Education of the Russian Federation
(agreement 075-15-2025-344 dated 29/04/2025 for Saint Petersburg
Leonhard Euler International Mathematical Institute at PDMI RAS)}.

\subsubsection*{\blue Manifold and coordinates}

\noindent$\bullet$\,\,\, Let $(M,g)$ be a
smooth\footnote{Everywhere {\it smooth} means $C^\infty$-smooth.}
{\it three-dimensional} Riemannian manifold. By $\langle
a,b\rangle$ we denote the (point-wise) inner product of vector
fields $a,b$ (sections of $T\Omega$); $\Delta$ is the Beltrami -
Laplace operator on $M$. Let $\Omega\Subset M$ be a ball of radius
$T$ centered at a point $c\in\Omega$, $\Gamma:=\partial\Omega$. We
assume that $T<r^{\rm inj}_c$ (injectivity radius), so that ${\rm
exp}_c$ is smooth in $\Omega$.

Let $\tau(x):={\rm dist\,}(x,\Gamma)$,
$\Omega^\tau:=\{x\in\Omega\,|\,\,\tau(x)<\tau\}$,\,\,
$\Gamma^\tau:=\{x\in\Omega\,|\,\,\tau(x)=\tau\}$,\,\,$0\leqslant\tau\leqslant
T$. The semi-geodesic coordinates (s.g.c.) $x\mapsto
(\gamma,\tau):\,\,\tau=\tau(x),\gamma=\gamma(x)$, where
$\gamma(x)\in\Gamma$ satisfies ${\rm
dist\,}(x,\gamma(x))=\tau(x)$, are defined and regular in
$\dot\Omega:=\Omega\setminus\{c\}$. By
$x(\gamma,\tau)\in\dot\Omega$ we denote the point with the s.g.c.
$(\gamma,\tau)\in\Gamma\times[0,T)$. Such an $\Omega$ is
diffeomorphic to a ball in $\Bbb R^3$.

Let $\gamma^1,\gamma^2$ be the local coordinates on $\Gamma$ near
$\gamma(x)$, and
$\partial_\tau,\partial_{\gamma^1},\partial_{\gamma^2}$ be the
coordinate fields in $\dot\Omega$. The length and volume elements
are
$$
ds^2=d\tau+g_{ij}(\gamma,\tau)d\gamma^i d\gamma^j;\qquad
dx=g^{\frac{1}{2}}(\gamma,\tau)\,d\tau\, d\gamma^1 d\gamma^2,\quad
g:={\rm det\,}\{g_{ij}\},
$$

In $\dot\Omega$, each vector field $h$ is represented as
\begin{equation}\label{Eq h=h lambda+h theta}
h=h_\lambda+h_\theta;\quad h_\lambda:=\langle
h,\partial_\tau\rangle\partial_\tau,\quad h_\theta:=h-h_\lambda
\end{equation}
with the components $h_\lambda$ and $h_\theta$, which are called
longitudinal (parallel to $\partial_\tau$) and transversal parts
of $h$, whereas $\langle h_\lambda(x),h_\theta(x)\rangle=0$,
$x\in\dot\Omega$ holds. Note that
$\nu:=-\partial_\tau\big|_\Gamma$ is the field of the outward
normals at the boundary.

\noindent$\bullet$\,\,\, The vector analysis operations (see
\cite{Sch}) in s.g.c. are
\begin{align}\label{Eq def nabla div in sgc}
\nabla p=(\partial_\tau
p)\,\partial_\tau\,+\,\left(g^{ij}{\partial_{\gamma^j}
p}\right)\,\partial_{\gamma^i}; \,\,{\rm div\,} h=g^{-\frac{1}{2}}\partial_\tau (g^{\frac{1}{2}}h^0)+g^{-\frac{1}{2}}\partial_{\gamma^i} (g^{\frac{1}{2}}h^i)
\end{align}
in $\dot\Omega$, where $\{g^{ij}\}:=\{g_{ij}\}^{-1}$ and
$h=h^0\partial_\tau+h^i\partial_{\gamma^i}$ is a smooth vector
field. We say that
$\nabla_\lambda:=(\partial_\tau\cdot)\,\partial_\tau$ and
$\nabla_\theta:=(g^{ij}\partial_{\gamma^j}\cdot)\,\partial_{\gamma^i}$
are the longitudinal and transversal parts of the gradient. The
operations ${\rm div}_\lambda:=g^{-\frac{1}{2}}\partial_\tau
(g^{\frac{1}{2}}\langle\cdot,\partial_\tau\rangle)$ and ${\rm
div}_\theta:=g^{-\frac{1}{2}}\partial_{\gamma^i} (g^{\frac{1}{2}}\langle\cdot,\partial_{\gamma^i}\rangle$ are the longitudinal
and transversal parts of the divergence; so, one represents
\begin{equation}\label{Eq div=div lambda+ div theta}
\nabla p\,=\,\nabla_\lambda p+\nabla_\theta p,\quad{\rm div\,}
h={\rm div}_\lambda h_\lambda+{\rm div}_\theta h_\theta\qquad {\rm
in}\,\,\,\dot\Omega
\end{equation}
in accordance with (\ref{Eq h=h lambda+h theta}).

\subsubsection*{\blue Spaces, subspaces and projections}

\noindent$\bullet$\,\,\, Let ${\mathscr L}$ be the space of the
square integrable  vector fields,
\begin{align}
\notag &(a,b)_{\mathscr L}:=\int_{\Omega}\langle
a,b\rangle\,dx=\int_{\Gamma\times[0,T)}
\left[a^0b^0+g_{ij}a^ib^j\right]\,g^{\frac{1}{2}}\,\,d\gamma^1d\gamma^2\,
d\tau\,=\\
\label{Eq inner product} &=
\int_0^Td\tau\int_{\Gamma^\tau}d\Gamma^\tau\, \langle
a,b\rangle\big|_{\Gamma^\tau},
\end{align}
where $a=a^0\partial_\tau+a^i\partial_{\gamma^i}$,
$b=b^0\partial_\tau+b^i\partial_{\gamma^i}$, and
$d\Gamma^\tau=g^{\frac{1}{2}}(\gamma,\tau)\,d\gamma^1d\gamma^2=\left[\frac{g(\gamma,\tau)}{g(\gamma,0)}\right]^{\frac{1}{2}}d\Gamma$ is the surface element on $\Gamma^\tau$, $d\Gamma$ is
the surface element on $\Gamma$. The latter equality in (\ref{Eq
inner product}) corresponds to the representations
\begin{equation}\label{Eq mL=oplus int}
\Omega=\bigcup_{0\leqslant\tau\leqslant T}\Gamma^\tau,\qquad
L_2(\Omega)\,=\,\oplus\int_{[0,T]}L_2(\Gamma^\tau)\,d\tau.
\end{equation}
Recall that our $\Omega$ is a Riemannian ball diffeomorphic to a
ball in $\Bbb R^3$. In such a case, the Helmholtz - Weyl
decomposition on the potential and solenoidal fields is
\begin{equation}\label{Eq Helm-Weyl}
{\mathscr L}={\mathscr P}\oplus\mathscr S_0,
\end{equation}
where
\begin{align*}
\notag & {\mathscr P}:=\{\nabla p\,|\,\,p\in H^1(\Omega)\},\quad
\mathscr S_0:=\{h\in{\mathscr L}\,|\,\,{\rm div\,}h=0\,\,{\rm
in}\,\,\Omega,\,\,\langle
h,\nu\rangle=0\,\,{\rm on}\,\,\Gamma\}=\\
& =\overline{\{h\in{\mathscr L}\,|\,\,{\rm div\,}h=0\,\,{\rm
in}\,\,\Omega,\,\,{\rm supp\,}h\subset\Omega\setminus\Gamma\}}.
\end{align*}
The first summand is specified as follows:
\begin{equation}\label{Eq Harm subspace}
{\mathscr P}={\mathscr P}_0\oplus{\mathscr H},
\end{equation}
where ${\mathscr P}_0:=\{\nabla p\,|\,\,p\in
H^1_0(\Omega)\}=\overline{\{\nabla p\,|\,\,{\rm
supp\,}p\subset\Omega\setminus\Gamma}\}$\,\, and ${\mathscr
H}:=\{\nabla p\,|\,\,\Delta p=0\,\,\,{\rm in}\,\,\,\Omega\}$ (see,
e.g., \cite{Sch}). We say ${\mathscr H}$ to be the subspace of
harmonic potential fields.
\smallskip

\noindent$\bullet$\,\,\, One more decomposition, which corresponds
to (\ref{Eq h=h lambda+h theta}), is
\begin{equation*}
{\mathscr L}={\mathscr L}_\lambda\oplus{\mathscr L}_\theta,
\end{equation*}
where ${\mathscr L}_\lambda:=\{h\in{\mathscr L}\,|\,h_\theta=0\}$
and ${\mathscr L}_\theta:=\{h\in{\mathscr L}\,|\,h_\lambda=0\}$
are the subspaces of the longitudinal and transversal fields
respectively.
\smallskip

\noindent$\bullet$\,\,\, Denote
$$
{\mathscr P}^\tau:=\{h\in{\mathscr P}\,|\,\,{\rm
supp\,}h\subset\overline{\Omega^\tau}\}, \quad 0<\tau\leqslant
T;\quad {\mathscr P}^0:=\{0\},\,\,{\mathscr P}^T=\mathscr P,
$$
and ${\mathscr P}^\tau_\bot:={\mathscr P}\ominus{\mathscr
P}^\tau$. For a field $h=\nabla p\in{\mathscr P}^\tau$ we have
$\nabla p=0$ in $\Omega\setminus\Omega^\tau$, so that $p=\rm
const$ outside $\Omega^\tau$ holds, and we put
$p\big|_{\Omega\setminus\Omega^\tau}=0$. By $P^\tau$ and
$P^\tau_\bot$ we denote the projections in ${\mathscr P}$ onto the
first and second summands in the decomposition ${\mathscr
P}={\mathscr P}^\tau\oplus{\mathscr P}^\tau_\bot$.
\begin{Lemma}\label{L decomp mP=mP tau+mP tau bot}
For a smooth $h=\nabla p\in{\mathscr P}$, the representations
$P^\tau h=\nabla p^\tau$ and $P^\tau_\bot h=\nabla p^\tau_\bot$
hold with the potentials satisfying
\begin{equation}\label{Eq p=pt+ptb}
\begin{cases}
\Delta p^\tau={\rm div\,}h & {\rm in}\,\,{\rm int\,}\Omega^\tau;\\
p^\tau=0 & {\rm on}\,\,\Gamma^\tau;\\
\partial_\tau p^\tau=\partial_\tau p & {\rm on}\,\,\Gamma;
\end{cases}\qquad,\qquad
\begin{cases}
\Delta p^\tau_\bot=0 & {\rm in}\,\,{\rm int\,}\Omega^\tau;\\
p^\tau_\bot=p & {\rm on}\,\,\Gamma^\tau;\\
\partial_\tau  p^\tau_\bot=0 & {\rm on}\,\,\Gamma;
\end{cases}
\end{equation}
and $p^\tau\big|_{\Omega\setminus\Omega^\tau}=0$,
$p^\tau_\bot\big|_{\Omega\setminus\Omega^\tau}=p$.
\end{Lemma}
\begin{proof}
As is easy to check, the relations $h=\nabla p=\nabla
p^\tau+\nabla p^\tau_\bot$, $\nabla p^\tau\in {\mathscr P}^\tau$
and $(\nabla p^\tau,\nabla p^\tau_\bot)_{\mathscr L}=0$ hold,
which is equivalent to the statement of the Lemma.
\end{proof}
So, the decomposition $h=P^\tau h+P^\tau_\bot h$ takes the form
$\nabla p=\nabla p^\tau+\nabla p^\tau_\bot$ with the potentials
obeying (\ref{Eq p=pt+ptb}).

\subsubsection*{\blue N-transform}

\noindent$\bullet$\,\,\, Fix $\tau\in (0,T)$ and a smooth
$h=\nabla p\in{\mathscr P}$. The field $P^\tau
h\big|_{\Gamma^\tau}:=P^\tau h\big|_{\Gamma^{\tau-0}}$ is
supported on $\Gamma^\tau$ and is orthogonal to $\Gamma^\tau$.
Indeed, in view of $p=p^\tau_\bot$ on $\Gamma^\tau$, the equality
$\nabla_\theta p=\nabla_\theta\, p^\tau_\bot$ holds on
$\Gamma^\tau$, and for $P^\tau h=h-P^\tau_\bot h$ we have
\begin{align}
\notag & P^\tau h\big|_{\Gamma^{\tau}}=\left(\nabla p-\nabla
p^\tau_\bot\right)\big|_{\Gamma^{\tau}}=\left(\partial_\tau
p\,\partial_\tau+\nabla_\theta
p-\partial_\tau p^\tau_\bot\,\partial_\tau-\nabla_\theta p^\tau_\bot\right)\big|_{\Gamma^{\tau}}=\\
\label{Eq auxil 1}& = \left(\partial_\tau
p\,\partial_\tau-\partial_\tau
 p^\tau_\bot\,\partial_\tau\right)\big|_{\Gamma^{\tau}}=\left(\partial_\tau
p-\partial_\tau
p^\tau_\bot\right)\,\partial_\tau\big|_{\Gamma^\tau}.
\end{align}

Define $N:{\mathscr P}\to{\mathscr L}_\lambda$\, layer-wise, i.e.,
in accordance with (\ref{Eq mL=oplus int}), by
\begin{equation}\label{Eq def N}
Nh\big|_{\Gamma^{\tau}}\,:=\,P^\tau h\big|_{\Gamma^{\tau}},\qquad
0\leqslant\tau\leqslant T
\end{equation}
(recall that $P^0=\Bbb O$ and $P^T=\Bbb I$) on smooth fields $h$.
As is shown in \cite{B MN transform}, the transform $N$ is an
isometry, which is extended by continuity from smooth $h$'s onto
${\mathscr P}$ to a unitary operator from ${\mathscr P}$ onto
${\mathscr L}_\lambda$. In Appendix, we provide a proof of these
properties, which in fact basically repeats the proof in \cite{B
MN transform}.
\smallskip

\noindent$\bullet$\,\,\, For $0<\tau<T$, we define the operator
$\Pi^\tau:L_2(\Gamma^\tau)\to L_2(\Gamma^\tau)$ on ${\rm
Dom\,}\Pi^\tau=H^1(\Gamma^\tau)$ by $\Pi^\tau f:=\partial_\tau
u^f$, where $u^f$ is a solution to
\begin{equation}
\label{additional number}
\begin{cases}
\Delta u=0 & {\rm in}\,\,{\rm int\,}\Omega^\tau;\\
u=f & {\rm on}\,\,\Gamma^\tau;\\
\partial_\tau u=0 & {\rm on}\,\,\Gamma.
\end{cases}
\end{equation}
The following facts are well known.
\begin{Proposition}\label{P Pi tau>O}
Operator $\Pi^\tau$ is a 1-st order PDO with the pricipal symbol
$\sigma^\tau(k)=[g^{ij}\big|_{\Gamma^\tau}k_ik_j]^{\frac{1}{2}}$. The
relations $\Pi^\tau=\Pi^{\tau\,*}\geqslant\Bbb O$, ${\rm
Ker\,}\Pi^\tau=\{\rm const\}$ hold.
\end{Proposition}
In the space $L_2(\Omega)$ (see (\ref{Eq mL=oplus int})) we define
the layer-wise operator $\Pi$ on ${\rm
Dom\,}\Pi=\oplus\int_{[0,T]}H^1(\Gamma^\tau)\,d\tau$ by
$$\
(\Pi\varkappa)\big|_{\Gamma^\tau}\,:=\,\Pi^\tau(\varkappa\big|_{\Gamma^\tau}),\qquad
0<\tau<T.
$$
One can easily verify the following facts.
\begin{Proposition}\label{P Pi>O}
The relations $\Pi=\Pi^{*}\geqslant\Bbb O$, ${\rm
Ker\,}\Pi=\mathscr C$,\, ${\rm Ran\,}\Pi\,\bot\mathscr C$ hold,
where $\mathscr C$ is the class of the square-summable leyer-wise
constant functions.
\end{Proposition}

Taking into account (\ref{Eq auxil 1}) and (\ref{Eq def N}), we
represent
\begin{equation}\label{Eq repres N via Pi}
Nh=N\nabla p=\left(\partial_\tau p-\Pi p\right)\,\partial_\tau.
\end{equation}

\noindent$\bullet$\,\,\, For a fixed $\tau\in[0,T)$, $\Gamma^\tau$
is a 2-dimensional Riemannian manifold (surface) with the metric
$g^\tau=g\big|_{\Gamma^\tau}$ induced by the metric $g$ in
$\Omega$. The elements of the vector field space ${\mathscr
L}^\tau:=\vec L_2(\Gamma^\tau)$ can be identified with the traces
of transversal fields in $\Omega$: ${\mathscr
L}^\tau=\overline{\{a\big|_{\Gamma^\tau}\,|\,\,a\,\,\,\text{is
smooth},\,\,a\in {\mathscr L}_\theta\}}$. It contains the subspace
${\mathscr P}_{\Gamma^\tau}:=\{\nabla^\tau \phi\,|\,\,\phi\in
H^1(\Gamma^\tau)\}$ of the potential fields, where $\nabla^\tau$
is the intrinsic gradient on $\Gamma^\tau$.

Put
\begin{equation}\label{Eq auxil 0}
\dot L_2(\Gamma^\tau):=\{\phi\in
L_2(\Gamma^\tau)\,|\,\,(\phi,1)=\int_{\Gamma^\tau}\phi\,d\Gamma^\tau=0\},\qquad
0\leqslant\tau<T.
\end{equation}
The vector analysis operations $\nabla^\tau:\dot
L_2(\Gamma^\tau)\to{\mathscr P}_{\Gamma^\tau}$ and ${\rm
div\,}^\tau=-\nabla^{\tau\,*}:{\mathscr P}_{\Gamma^\tau}\to\dot
L_2(\Gamma^\tau)$ are injective.

Introduce the spaces
$$
\dot L_2(\Omega)\,:=\,\oplus\int_{[0,T]}\dot
L_2(\Gamma^\tau)\,d\tau=L_2(\Omega)\ominus\mathscr
C\,=\,\overline{{\rm Ran\,}\Pi}
$$
and $\mathscr Q_\theta\,:=\,\oplus\int_{[0,T]} {\mathscr
P}_{\Gamma^\tau}\,d\tau\subset{\mathscr L}_\theta$. The layer-wise
operations
$$
\nabla_\theta\big|_{\dot
L_2(\Omega)}=\oplus\int_{[0,T]}\nabla^\tau\,d\tau:\dot
L_2(\Omega)\to\mathscr Q_\theta
$$
and
$$
{\rm div\,}_\theta\big|_{\mathscr
Q_\theta}=\oplus\int_{[0,T]}{\rm div\,}^\tau\,d\tau:\mathscr
Q_\theta\to\dot L_2(\Omega)
$$
are injective. Therefore, the operations
$\nabla_\theta^{-1}:\mathscr Q_\theta\to\dot L_2(\Omega)$ and
${\rm div\,}_\theta^{-1}:\dot L_2(\Omega)\to\mathscr Q_\theta$ are
well defined, and the relations
\begin{equation}\label{Eq auxil 2}
\nabla_\theta^{-1}\nabla_\theta={\rm id},\qquad
(\nabla_\theta^{-1})^*=-\,{\rm div\,}_\theta^{-1}
\end{equation}
are valid on $\dot L_2(\Omega)$ and $\mathscr Q_\theta$
respectively.
\smallskip

\noindent$\bullet$\,\,\,The latter relations are used to derive
the following representation.
\begin{Lemma}\label{L repres N^*}
Let $\varkappa$ be smooth in $\dot\Omega$,
$v=\varkappa\,\partial_\tau\in{\mathscr L}_\lambda$. For
$N^*:{\mathscr L}_\lambda\to {\mathscr P}$ the representation
\begin{equation}\label{Eq N^*}
N^*v\,=\,P_{\mathscr P}\left[\varkappa\,\partial_\tau\,+\,{\rm
div\,}_\theta^{-1}\Pi\varkappa\right]
\end{equation}
is valid, where $P_{\mathscr P}$ is the projection onto ${\mathscr
P}$ in (\ref{Eq Helm-Weyl}). The relation
\begin{equation}\label{Eq div N^*}
{\rm div\,} N^*v\,=\,{\rm
div\,}\varkappa\,\partial_\tau+\Pi\varkappa
\end{equation}
holds.
\end{Lemma}
\begin{proof} Let $p$ and $\varkappa$ be smooth, $h=\nabla p\in{\mathscr P}$ and
$v=\varkappa\partial_\tau\in{\mathscr L}_\lambda$. Let $\dot p$ be
the projection in $L_2(\Omega)$ onto $\dot
L_2(\Omega)=\overline{{\rm Ran\,}\Pi}$, so that $\nabla_\theta
p=\nabla_\theta\dot p$ holds. Then we have
\begin{align*}
& (Nh,v)_{\mathscr L}\overset{\text{see}\,\,(\ref{Eq repres N via
Pi})}=(\partial_\tau
p\,\partial_\tau,\varkappa\,\partial_\tau)_{{\mathscr L}}-(\Pi p\,\partial_\tau,\varkappa\,\partial_\tau)_{{\mathscr L}}=\\
& =(\partial_\tau p,\varkappa)_{L_2(\Omega)}-(\Pi
p,\varkappa)_{L_2(\Omega)}=(\partial_\tau
p\,\partial_\tau,\varkappa\partial_\tau)_{{\mathscr
L}_\lambda}-(\dot
p,\Pi\varkappa)_{L_2(\Omega)}=\\
& \overset{(\ref{Eq auxil 2})}=(\partial_\tau
p\,\partial_\tau,\varkappa\partial_\tau)_{{\mathscr
L}_\lambda}-(\nabla_\theta^{-1}\nabla_\theta
p,\Pi\varkappa)_{L_2(\Omega)}\overset{(\ref{Eq auxil
2})}=(\partial_\tau
p\,\partial_\tau,\varkappa\partial_\tau)_{{\mathscr
L}_\lambda}+(\nabla_\theta p,{\rm
div\,}_\theta^{-1}\Pi\varkappa)_{{\mathscr L}_\theta}=\\
& =(\partial_\tau p\,\partial_\tau+\nabla_\theta
p,\varkappa\partial_\tau+{\rm
div\,}_\theta^{-1}\Pi\varkappa)_{{\mathscr L}}=(\nabla
p,\varkappa\partial_\tau+{\rm
div\,}_\theta^{-1}\Pi\varkappa)_{{\mathscr L}}=\\
& =(h,\varkappa\partial_\tau+{\rm
div\,}_\theta^{-1}\Pi\varkappa)_{{\mathscr L}} =(P_{\mathscr P}
h,\varkappa\partial_\tau+{\rm
div\,}_\theta^{-1}\Pi\varkappa)_{{\mathscr L}}=\\
& =( h,P_{\mathscr P}[\varkappa\partial_\tau+{\rm
div\,}_\theta^{-1}\Pi\varkappa])_{{\mathscr L}}=(h,N^*v)_{\mathscr
L}.
\end{align*}
So, by the density of smooth $h$'s in $\mathscr P$, for smooth
$\varkappa$'s the representation (\ref{Eq N^*}) does hold. By the
simply verified boubdedness of the composition ${\rm
div\,}_\theta^{-1}\Pi\varkappa]$ it can be extended to $\mathscr
L_\lambda$ by continuity.

Denoting $a:=\varkappa\partial_\tau+{\rm
div\,}_\theta^{-1}\Pi\varkappa$, we have $a=P_{\mathscr P}
a+P_\mathscr S a$ (see (\ref{Eq Helm-Weyl})) and $P_{\mathscr P}
a=a-P_\mathscr S a$, ${\rm div\,}\,P_\mathscr S a=0$, which
implies
\begin{align*}
& {\rm div\,} N^*v={\rm div\,} P_{\mathscr P} a={\rm div\,} a
-{\rm div\,} P_\mathscr S a={\rm div\,} a={\rm
div\,}\,\varkappa\partial_\tau+{\rm
div\,}_\theta^{-1}\Pi\varkappa=\\
&\overset{(\ref{Eq div=div lambda+ div theta})}={\rm
div\,}\varkappa\partial_\tau+{\rm div\,}_\theta\,{\rm
div\,}_\theta^{-1}\Pi\varkappa={\rm
div\,}\,\varkappa\partial_\tau+\Pi\varkappa.
\end{align*}
\end{proof}

\noindent$\bullet$\,\,\, Let $h=\nabla p$ and ${\rm div\,}\,h=0$
hold, so that $\Delta p=0$ holds, i.e., the potential $p$ is
harmonic in $\Omega$. Let $v=\varkappa\,\partial_\tau=Nh$. Then
${\rm div\,}N^*v={\rm div\,} N^*Nh={\rm div\,} h=\Delta p=0$ holds
and, by (\ref{Eq div N^*}), implies
\begin{equation}\label{Eq varkappa harm}
{\rm div\,}\varkappa\,\partial_\tau+\Pi\varkappa\,=\,0\qquad {\rm
in}\,\,\,{\rm int\,}\Omega.
\end{equation}

\subsection*{\blue Factorization of DN-map}

\noindent$\bullet$\,\,\, Consider the Dirichlet problem
$$
\begin{cases}
\Delta p=0 &{\rm in}\,\,\,{\rm int\,}\Omega;\\
p=f &{\rm on}\,\,\,\Gamma;
\end{cases}
$$
let $p=p^f(x)$ be a solution, ${\rm div\,}\nabla p^f=\Delta p^f=0$
holds in $\Omega$. The operator $W:L_2(\Gamma)\to\mathscr
H\subset{\mathscr P}$ (see (\ref{Eq Harm subspace})), $Wf:=\nabla
p^f$ resolves the problem.

The DN map is $\Lambda:L_2(\Gamma)\to {\rm L}_2(\Gamma)$, ${\rm
Dom\,}\Lambda=H^1(\Gamma)$,
$$
\Lambda f\,:=\,\nu p^f\qquad {\rm on}\,\,\Gamma,
$$
where $\nu=-\partial_\tau\big|_{\Gamma}$ is the outward normal at
the boundary. Since
\begin{align}\label{Eq def Lambda}
\notag & (\nabla p^f,\nabla p^g)_{\mathscr H}=(Wf,Wg)_{\mathscr
H}=\int_\Omega \langle
\nabla p^f,\nabla p^g\rangle\,dx=\\
& =\int_\Gamma \nu p^f\,p^g\,d\Gamma=\int_\Gamma \Lambda
f\,g\,d\Gamma,
\end{align}
we have $\Lambda=\Lambda^*=W^*W\geqslant\Bbb O$, ${\rm
Ran\,}\Lambda=\dot L_2(\Gamma)$ (see (\ref{Eq auxil 0})) and ${\rm
Ker\,}\Lambda=\{\rm const\}$.
\smallskip

\noindent$\bullet$\,\,\, Let $N\nabla
p^f=\varkappa^f\,\partial_\tau\in{\mathscr L}_\lambda$. Then ${\rm
div\,} \varkappa^f\partial_\tau+\Pi\varkappa^f\overset{(\ref{Eq
varkappa harm})}=0$ holds and takes the form $g^{-\frac{1}{2}}\partial_\tau(g^{\frac{1}{2}}\varkappa^f)+\Pi\varkappa^f=0$ in
s.g.c. (see (\ref{Eq def nabla div in sgc}) and (\ref{Eq div=div
lambda+ div theta})). Also, in view of $\partial_\tau
p^\tau_\bot\big|_\Gamma\overset{(\ref{Eq p=pt+ptb})}=0$, we have
$$
N\nabla p^f\big|_\Gamma\overset{(\ref{Eq auxil 1})}=\partial_\tau
p^f\partial_\tau\big|_\Gamma=-\nu
p^f\partial_\tau\big|_\Gamma=-\Lambda
f\,\partial_\tau\big|_\Gamma=(\varkappa^f\,\partial_\tau)\big|_\Gamma,
$$
which implies $\varkappa^f=-\Lambda f$ on $\Gamma$. As a result,
$\varkappa=\varkappa^f(x(\gamma,\tau))$ satisfies
\begin{equation}\label{Eq auxil 3}
\begin{cases}
g^{-\frac{1}{2}}\partial_\tau(g^{\frac{1}{2}}\varkappa)+\Pi\varkappa=0 &{\rm in}\,\,\,\Gamma\times(0,T);\\
\varkappa\big|_{\tau=0}=-\Lambda f;
\end{cases}.
\end{equation}
By $V$ we denote the operator (propagator)
$V:\varkappa\big|_{\Gamma}\mapsto
\varkappa\big|_{\Gamma\times(0,T)}$ that resolves the Cauchy
problem (\ref{Eq auxil 3}) \cite{Tsutsumi,Krainer,Yosida}. Since
${\rm Ran\,}\Lambda=\dot L_2(\Gamma)$ holds, $V$ acts from $\dot
L_2(\Gamma)$ to $N\mathscr H\subset \mathscr L_\lambda$.

\noindent$\bullet$\,\,\, By isometry of the N-transform and
(\ref{Eq def Lambda}) we have
\begin{align*}
\notag & (\Lambda f\,g)_{L_2(\Gamma)}=(\nabla p^f,\nabla
p^g)_{\mathscr H}=(N\nabla p^f,N\nabla
p^g)_{{\mathscr L}_\lambda}=(\varkappa^f\partial_\tau,\varkappa^g\partial_\tau)_{{\mathscr L}_\lambda}=\\
\notag &
=(\varkappa^f,\varkappa^g)_{L_2(\Omega)}=(V\varkappa^f\big|_{\tau=0},V\varkappa^g\big|_{\tau=0})_{L_2(\Omega)}=(V\Lambda
f,V\Lambda g)_{L_2(\Omega)}=\\
& = (\Lambda V^*V\Lambda f,g)_{L_2(\Gamma)},
\end{align*}
which implies $\Lambda=\Lambda V^*V\Lambda$ and leads to
\begin{equation}\label{Eq Lambda factoriz+}
\Lambda^{-1}\,=\,V^*V.
\end{equation}
so that $V:\dot L_2(\Gamma)\to N\mathscr H$ provides a
factorization to $\Lambda^{-1}$.

\subsubsection*{\blue Illustration: upper half-space}

\noindent$\bullet$\,\,\, Let $\Omega=\overline{\Bbb
R^3_+}=\{(x,z)\,|\,\,x=(x^1,x^2)\in\Bbb R^2,\,\,z\geqslant 0\}$,
$\Gamma=\{(x,0)\,|\,\,x\in\Bbb R^2\}$. A harmonic potential
$p=p^f$ satisfies
\begin{equation*}
\begin{cases}
\Delta p=0 & {\rm in}\,\,{\rm int\,}\Omega;\\
p=f & {\rm on}\,\,\Gamma;\\
p\overset{}\to 0 & {\rm as}\,\,|x|+z\to\infty;
\end{cases}.
\end{equation*}
The Fourier transform
$$
\tilde u(k,z)=(Fu(\cdot,z))(k):=(2\pi)^{-\frac{3}{2}}\int_{\Bbb
R^2}e^{-\langle k,x\rangle}u(x,z)\,dx,\qquad
$$
implies
\begin{equation*}
\begin{cases}
\tilde p_{zz}+|k|\,\tilde p=0, & z>0;\\
\tilde p=\tilde f, & z=0;\\
\tilde p\overset{}\to 0, & |k|+z\to\infty;
\end{cases}
\end{equation*}
and provides
$$
\tilde p\,(k,z)\,=\,e^{-|k|z}\tilde f(k),\qquad k\in \Bbb
R^2_*,\,\,z\geqslant 0,
$$
whereas $\tilde\Lambda :=F\Lambda F^*$ \,just multiplies functions
by $|k|$.
\smallskip

\noindent$\bullet$\,\,\, The second problem in (\ref{Eq p=pt+ptb})
takes the form
\begin{equation*}
\begin{cases}
(\tilde p^\tau_\bot)_{zz}+|k|^2(\tilde p^\tau_\bot)=0, & z>0;\\
(\tilde p^\tau_\bot) =\tilde p, & z=\tau;\\
(\tilde p^\tau_\bot)_z=0, & z=0;
\end{cases}
\end{equation*}
and provides
$$
\tilde p^\tau_\bot(k,z)=\frac{\cosh |k|z}{\cosh |k|\tau}\,\tilde
p(k,z),\qquad 0\leqslant z\leqslant \tau.
$$
Hence we have
$$
(\tilde p^\tau_\bot)_{z}(k,\tau)=\left(\Pi^\tau \tilde
p(\cdot.\tau)\right)(k)=(|k|\,\tanh|k|\tau)\,\tilde
p(k,\tau),\qquad \tau\geqslant 0.
$$

\noindent$\bullet$\,\,\, Problem (\ref{Eq auxil 3}) takes the form
\begin{equation}\label{Eq auxil 4}
\begin{cases}
\tilde\varkappa_\tau+(|k|\,\tanh|k|\tau)\tilde\varkappa=0,&\tau>0;\\
\tilde\varkappa=-\widetilde{\Lambda f}=-\tilde\Lambda\tilde
f=-|k|\tilde f, & z=0;
\end{cases}
\end{equation}
the solution is
$$
\tilde\varkappa(k,\tau)= e^{-\int_0^\tau
|k|\,\tanh|k|\eta\,d\eta}\,\tilde\varkappa(k,0)=\frac{\tilde\varkappa(k,0)}{\cosh|k|\tau}=-\,\frac{|k|}{\cosh|k|\tau}\,\tilde
f(k).
$$
Thus, the operator resolving (\ref{Eq auxil 4}) acts from
$L_2(\{k\in \Bbb R^2_*\})$ to $L_2\left(\Bbb R^2_*\times
\{z\geqslant 0\}\right)$ by
\begin{equation*}
(\tilde
V\tilde\varkappa_0)(k,z)\,=\,\frac{\tilde\varkappa_0(k)}{\cosh|k|z},
\qquad k\in\Bbb R^2_*,\,\,z\geqslant 0.
\end{equation*}

As is easy to check, for $\tilde V^*:L_2\left(\Bbb R^2_*\times
\{z\geqslant 0\}\right)\to L_2(\{k\in \Bbb R^2_*\})$ one has
$$
(\tilde V^*
\phi)(k)\,=\,\int_0^\infty\frac{\phi(k,z)}{\cosh|k|z}\,dz,\qquad
k\in\Bbb R^2_*.
$$
As a result, we get
$$
(\tilde V^*\tilde V\tilde
\varkappa_0)(k)\,=\,\int_0^\infty\frac{dz}{{\cosh|k|z}}\frac{\tilde
\varkappa_0(k)}{{\cosh|k|z}}=\tilde\varkappa_0(k)\int_0^\infty\frac{dz}{\cosh^2|k|z}=\frac{\tilde\varkappa_0(k)}{|k|}.
$$
Thus, we have $\tilde V^*\tilde V={\tilde\Lambda}^{-1}$ in the
accordance with (\ref{Eq Lambda factoriz+}).

\subsubsection*{Illustration: unit ball}
Let $\Omega=\mathbb{B}=:\{x\in\mathbb{R}^3 \ | \ |x|\le 1\}$; then
$\Gamma^\tau:=\{x\in\mathbb{R}^3 \ | \ |x|= 1-\tau\}$,
$\Gamma=\Gamma^0=\Bbb S^2$. Introduce the spherical coordinates
$(r,\vartheta,\varphi)$; then $\partial_\tau=-\partial_r$. Let
$u^{\tau}_{lm}$ be a solution to (\ref{additional number}), where
$f(\vartheta,\varphi)=Y_l^m(\vartheta,\varphi)$ is a spherical
harmonic. As is easy to derive,
$$
u^{\tau}_{lm}=\frac{(l+1)r^l+lr^{-(l+1)}}{(l+1)\rho^l+l\rho^{-(l+1)}}Y_l^m(\vartheta,\varphi)
\qquad (\rho:=1-\tau)
$$
holds. Hence, one has
$$
\Pi^\tau Y_l^m=-\partial_r u^{\tau}_{lm}|_{r=\rho}=\frac{l(l+1)}{\rho}\frac{1-\rho^{2l+1}}{l+(l+1)\rho^{2l+1}}Y_l^m(\vartheta,\varphi).
$$
In particular, if $f=Y^{m}_l$ on $\Bbb S^2$, then $p^f=r^l f$,
$\Lambda f=lf$, and one obtains
$$
N\nabla p^f=-l\,\lambda_l(\rho)\,f\,\partial_{\tau}; \qquad
Vf=\lambda_l(\rho)\,f, \qquad f=Y^{m}_l,
$$
where
$$
\lambda_l(\rho):=-\frac{(2l+1)\rho^{l-1}}{l+(l+1)\rho^{2l+1}}.
$$
Note that $N\nabla p^f$ ($f\in H^{1/2}(\Gamma)$) is bounded but
may have a jump discontinuity at the center of the ball.

Denote by $\mathfrak{P}_l$ the projection in $L_2(\Bbb S^2)$ on
the eqigenspace of the Laplace-Beltrami operator $\Delta_{\Bbb
S^2}$ corresponding to the eigenvalue $l(l+1)$. Then the above
formulas imply
$$[Vf](\cdot,\tau)=\sum_{l=1}^{\infty}\lambda_l(1-\tau)\,\mathfrak{P}_lf, \quad V^*\phi=\int_0^1 \sum_{l=1}^{\infty}[\mathfrak{P}_l\phi](\cdot,1-\rho)\lambda_l(\rho)\rho^2 d\rho.$$
At last, we have
$$
V^*Vf=\sum_{l=1}^{\infty}\Big(\int_0^1\lambda^2_l(\rho)\rho^2
d\rho\Big)\mathfrak{P}_lf=\sum_{l=1}^{\infty}l^{-1}\mathfrak{P}_lf=\Lambda^{-1}f
$$
in accordance with (\ref{Eq Lambda factoriz+}).

\subsubsection*{\blue Comments and hopes}

\noindent$\bullet$\,\,\,Note that, in problem (\ref{Eq auxil 3}),
the operator $\Pi$ is a layer-wise PDO, whereas its symbol
determines the metric $g$ in s.g.c. (Proposition \ref{P Pi
tau>O}). Consider the case in which $\Omega=\mathbb{R}^3_+$ is a
half-space endoved with the smooth metric
$$ds^2=(dx^3)^2+\sum_{ij=1,2}g_{ij}(x^1,x^2,x^3)\,dx^idx^j,$$
coinciding with the euclidean one outside a sufficiently large
ball $|x|<N$. Then the results of general theory of parabolic PDO
\cite{Tsutsumi,Krainer} applied to the initial problem (\ref{Eq
auxil 3}) imply that its evolution operator
$V^\tau:\,\varkappa(\cdot,0)\mapsto (V\varkappa)(\cdot,\tau)$ is a
negligible PDO for any $\tau>0$.
\smallskip

\noindent$\bullet$\,\,\,Perhaps the above could suggest an
approach to the 3D electrical impedance tomography problem. If we
characterized the factorization (\ref{Eq Lambda factoriz+}) so
that it was realizable (or at least unique), then we could recover
the metric $g$ in s.g.c. by the scheme $\Lambda\mapsto V\mapsto
\Pi\mapsto\Pi^\tau\mapsto g\big|_{\Gamma^\tau},\,\tau>0$\,\,\,(see
Proposition \ref{P Pi tau>O}).

Reducing the problem to the determination $\Lambda\mapsto V$, we
encounter a canonical situation: to recover an operator $V$ via
its module $|V|:=\sqrt{V^*V}=\Lambda^{-{1\over 2}}$, which is
quite common in inverse problems. In the time-domain IP's, the
property of $V$, due to that the factorization $\Lambda^{-1}=V^*V$
is realizable, is its triangularity, which reflects the
fundamental physical fact: the finiteness of the wave propagation
velocity \cite{BSim}. So, in EIT we need to recognize its relevant
"physical" analog. Perhaps, it is related to the complex
geometrical optics.

\section*{Appendix}
Recall that $\Omega$ is a ball diffeomorphic to a ball in $\Bbb
R^3$. The proof of the following basic property of the N-transform
in fact repeats the analogous proof in \cite{B MN transform}.
\begin{Theorem}
The N-transform is a unitary operator from $\mathscr P$ to
$\mathscr L_\lambda$.
\end{Theorem}
\begin{proof}
\noindent$\bullet$\,\,\,Show that the N-transform is an isometry.
Taking smooth $h=\nabla p$ and $g=\nabla q$, we have
\begin{align*}
& \frac{d}{d\tau}\,(P^\tau h,P^\tau g)_{\mathscr
L_\lambda}=\\
&=\frac{d}{d\tau}\,\int_{\Omega^\tau}dx\,\langle P^\tau h,P^\tau
g\rangle=\frac{d}{d\tau}\,\int_0^\tau
ds\int_{\Gamma^s}d\Gamma^s\,\langle h-\nabla p^\tau_\bot,g-\nabla
q^\tau_\bot\rangle\,=\,I^\tau\,+I\!I^\tau,
\end{align*}
where $I^\tau:=\int_{\Gamma^\tau}d\Gamma^\tau\,\langle h-\nabla
p^\tau_\bot,g-\nabla
q^\tau_\bot\rangle=\int_{\Gamma^\tau}d\Gamma^\tau\,\langle
Nh,Ng\rangle$, and
\begin{align*}
& I\!I^\tau:=\int_0^\tau
ds\int_{\Gamma^s}d\Gamma^s\,\left[\langle-\,\nabla\frac{\partial
p^\tau_\bot}{\partial\tau},\,{g-\nabla q^\tau_\bot}\rangle
+\langle h-\nabla p^\tau_\bot,\,-\,\nabla\frac{\partial
q^\tau_\bot}{\partial\tau}\rangle\right]=\\
& =\,\int_0^\tau
ds\int_{\Gamma^s}d\Gamma^s\,\left[\langle-\,\nabla\frac{\partial
p^\tau_\bot}{\partial\tau},\,\nabla(q-q^\tau_\bot)\rangle +\langle
\nabla (p-p^\tau_\bot),\,-\,\nabla\frac{\partial
q^\tau_\bot}{\partial\tau}\rangle\right]=\\
& =\,\,\int_{\Omega^\tau}dx\,\left[\Delta\frac{\partial
p^\tau_\bot}{\partial\tau}\,(q-q^\tau_\bot)
+(p-p^\tau_\bot)\,\Delta\frac{\partial
q^\tau_\bot}{\partial\tau}\right]-\\
& -\,\int_{\Gamma}d\Gamma\,\left[\frac{\partial (\partial_\nu
p^\tau_\bot)}{\partial\tau}\,(q-q^\tau_\bot)
+(p-p^\tau_\bot)\,\frac{\partial
(\partial_\nu q^\tau_\bot)}{\partial\tau}\right]-\\
& -\,\int_{\Gamma^\tau}d\Gamma^\tau\,\left[\frac{\partial
(\partial_\nu p^\tau_\bot)}{\partial\tau}\,(q-q^\tau_\bot)
+(p-p^\tau_\bot)\,\frac{\partial (\partial_\nu
q^\tau_\bot)}{\partial\tau}\right]=:\int_{\Omega^\tau}-\int_{\Gamma^\tau}-\int_{\Gamma},
\end{align*}
where $\nu$ is the outward normal to
$\partial\Omega^\tau=\Gamma\cup\Gamma^\tau$. The first equation of
the second system in (\ref{Eq p=pt+ptb}) easily provides
$\Delta\frac{\partial p^\tau_\bot}{\partial\tau}=\frac{\partial
\Delta p^\tau_\bot}{\partial\tau}=\Delta\frac{\partial
q^\tau_\bot}{\partial\tau}=\frac{\partial \Delta
q^\tau_\bot}{\partial\tau}=0$ in ${\rm int\,}\Omega^\tau$, so that
we have $\int_{\Omega^\tau}=0$. The second equation implies
$\int_{\Gamma^\tau}=0$, the third one leads to $\int_{\Gamma}=0$.
As a result, we get $I\!I^\tau=0$ and
$$
\frac{d}{d\tau}\,(P^\tau h,P^\tau g)_{\mathscr
L_\lambda}=\int_{\Gamma^\tau}d\Gamma^\tau\,\langle
Nh,Ng\rangle,\qquad 0<\tau<T.
$$

Integrating over $\tau$ with regard to $P^0=\Bbb O$ and $P^T=\Bbb
I$, we arrive at
$$
(h,g)_\mathscr
P\,=\,\int_0^Td\tau\int_{\Gamma^\tau}d\Gamma^\tau\,\langle
Nh,Ng\rangle=\int_{\Omega}dx\,\langle
Nh,Ng\rangle=(Nh,Ng)_{\mathscr L_\lambda},
$$
so that $N$ is an isometry on the smooth fields. Hence, its
extension by continuity (denoted by the same $N$) is an isometry
from $\mathscr P$ to $\mathscr L_\lambda$.
\smallskip

\noindent$\bullet$\,\,\,The following simple facts are used later.

\noindent{\bf 1.}\,\,\,Fix $0<\tau<T$. Any $h=\nabla p\in\mathscr
P^\tau$ is longitudinal on $\Gamma^\tau$ in view of
$p\big|_{\Gamma^\tau}={\rm const}=0$. By (\ref{Eq p=pt+ptb}), the
latter implies $p^\tau_\bot=0$ in $\Omega^\tau$ and, by (\ref{Eq
auxil 1}), follows to
\begin{equation}\label{Eq auxill 6}
Nh\,=\,h\qquad{\rm on}\,\,\Gamma^\tau.
\end{equation}

\noindent{\bf 2.}\,\,\,For any smooth $\psi$ given on
$\Gamma^\tau$\,\,($0<\tau<T$), there is a field $h\in\mathscr
P^\tau$ such that $h\big|_{\Gamma^{\tau-0}}=\psi\partial_\tau$
holds. Indeed, choosing in $\Omega^\tau$ a smooth function $u$
provided $u=0$ and $\partial_\tau u =\psi$ on $\Gamma^\tau$, and
putting $h\big|_{\Omega^\tau}=\nabla u$,
$h\big|_{\Omega\setminus\Omega^\tau}=0$, we get a required field.
\smallskip

\noindent$\bullet$\,\,\,Show that ${\rm Ran\,}N=\mathscr
L_\lambda$ holds.

The definition of the N-transform (\ref{Eq def N}) easily implies
\begin{equation*}
NP^\tau\,=\,X^\tau N,\qquad 0<\tau<T,
\end{equation*}
where $X^\tau$ cuts off fields on $\Omega^\tau$:
$$
X^\tau v\,=\,
\begin{cases}
v & {\rm in\,}\,\,\Omega^\tau;\\
0, & {\rm in\,}\,\,\Omega\setminus\Omega^\tau;
\end{cases}\,\,.
$$
As a consequence,  we have $P^\tau N^*=N^*X^\tau$, which leads to
\begin{equation}\label{Eq auxill 7}
X^\tau{\rm Ker\,}N^*\,\subset\,{\rm Ker\,}N^*, \qquad 0<\tau<T.
\end{equation}
By the latter, assuming $v=\varkappa \partial_\tau,\,\,v
\bot\,{\rm Ran\,}N$ or, equivalently, $v\in{\rm Ker\,}N^*$, we
have $X^\tau v\in {\rm Ker\,}N^*$ and
\begin{align*}
0\overset{(\ref{Eq auxill 7})}=(h,N^*X^\tau v)_\mathscr
P=(Nh,X^\tau v)_{\mathscr L_\lambda}=\int_{\Omega^\tau}\langle
Nh,v\rangle\,dx=\int_0^T ds\int_{\Gamma^\tau}d\Gamma^\tau\langle
Nh,v\rangle
\end{align*}
for $h\in\mathscr P$ and $0<\tau<T$. Differentiation provides
\begin{equation}\label{Eq auxill 8}
\int_{\Gamma^\tau}\langle Nh,v\rangle\,d\Gamma^\tau\,=\,0, \qquad
0<\tau<T.
\end{equation}

Fix a $\tau=\sigma$. Choose a smooth function $\psi$ on
$\Gamma^\tau$ and a field $h\in\mathscr P^\sigma$,
$h\big|_{\Gamma^\sigma}=\psi\partial_\tau$ that is possible due to
{\bf 2.}. For such a choice, one has
\begin{equation}\label{Eq auxill 9}
\int_{\Gamma^\sigma}\langle
Nh,v\rangle\,d\Gamma^\sigma\overset{(\ref{Eq auxill
6})}=\int_{\Gamma^\sigma}\langle
h,v\rangle\,d\Gamma^\sigma=\int_{\Gamma^\sigma}\psi\,\varkappa\,d\Gamma^\sigma\overset{(\ref{Eq
auxill 8})}=0.
\end{equation}
Since $\psi$ is arbitrary, (\ref{Eq auxill 9}) yields $\varkappa
=0$ on $\Gamma^\sigma$. Since $\sigma$ is arbitrary, we have
$\varkappa=0$ and hence $v=0$ in $\Omega$. Thus ${\rm
Ker\,}N^*=\mathscr L_\lambda\ominus{\rm Ran\,}N=\{0\}$ holds,
i.e., ${\rm Ran\,}N=\mathscr L_\lambda$ is valid.
\smallskip

\noindent$\bullet$\,\,\,So, the transform $N:\mathscr P\to\mathscr
L_\lambda$ is an isometry acting onto the image space, i.e., is a
unitary operator.
\end{proof}


\begin{thebibliography}{9}

\bibitem{B MN transform}
M.I.Belishev.
\newblock {On a unitary transform in the space $L_2(\Omega; \Bbb R^3)$ associated with the Weyl decomposition.}
\newblock {\em Journal of Mathematical Sciences}, 08/2003, Volume 117(2), pages
3900–3909, (2003). DOI:10.1023/A:1024606522660.

\bibitem{B IPI 2022}
M.I.Belishev.
\newblock {New Notions and Constructions of the Boundary
Control Method.}
\newblock {\em Inverse Problems and Imaging}, Vol. 16, No. 6, December 2022, pp. 1447-1471.
doi:10.3934/ipi.2022040.

\bibitem{BGlas_A&A_Maxwell}
M.I.Belishev, A.K.Glasman.
\newblock {Dynamical inverse problem for the Maxwell system: recovering
the velocity in a regular zone (the BC--method)}.
\newblock {\em St.-Petersburg Math. Journal}, 12(2): 279--316, 2001.

\bibitem{BSim}
M.I.Belishev, S.A.Simonov.
\newblock {Triangular and functional models of operators ad systems}.
\newblock {\em Algebra i Analiz}, 36:5 (2024), 101–127.\quad(in Russian)

\bibitem{Krainer}
T.Krainer, BW.Schulze.
\newblock {On the Inverse of Parabolic Systems of Partial Differential Equations of General Form in an Infinite Space-Time Cylinder}.
\newblock {In: S.Albeverio, M.Demuth, E.Schrohe, BW.Schulze. (eds) Parabolicity, Volterra Calculus, and Conical Singularities.}
\newblock {\em Operator Theory: Advances and Applications}, vol 138., Birkh\"auser, Basel (2002). \url{https://doi.org/10.1007/978-3-0348-8191-3_3}

\bibitem{Sch}
G.Schwarz.
\newblock {Hodge decomposition - a method for solving boundary value
problems}.
\newblock {\em Lecture notes in Math.}, 1607.
\newblock{\em Springer--Verlag, Berlin}, 1995.

\bibitem{Tsutsumi}
C.Tsutsumi.
\newblock {The fundamental solution for a parabolic pseudo-differential
operator and parametrices for degenerate operators}.
\newblock {\em Proc. Japan Acad.}, 51(2), pp. 103-108 (1975). DOI:
10.3792/pja/1195518695.

\bibitem{Yosida}
K.Yosida.
\newblock {Functional Analysis}.
\newblock{\em Springer--Verlag, Berlin, Goettingen, Heidelberg}, 1965.






\end{thebibliography}
\end{document}